\documentclass{CSML}

\def\dOi{11(3:23)2015}
\lmcsheading%
{\dOi}
{1--15}
{}
{}
{Jan.~13, 2014}
{Sep.~28, 2015}
{}

\ACMCCS{[{\bf Theory of computation}]: Logic; [{\bf Computing
    Methodologies}]: Symbolic and algebraic manipulation---Symbolic and algebraic algorithms---Theorem proving algorithms}

\usepackage{amssymb,color}
\usepackage{hyperref}

\newcommand{\f}{\ensuremath{\varphi}}
\newcommand{\p}{\ensuremath{\psi}}
\renewcommand{\a}{\ensuremath{\alpha}}
\renewcommand{\b}{\ensuremath{\beta}}
\newcommand{\si}{\sigma} 
\newcommand{\De}{\Delta}

\newcommand{\Si}{\Sigma}

\newcommand{\imp}{\Rightarrow}
\newcommand{\lang}{{\mathcal L}}
\newcommand{\alg}[1]{\mathbf{#1}}
\newcommand{\pos}[1]{\mathbf{#1}}
\newcommand{\A}{\alg{A}}
\newcommand{\B}{\alg{B}}
\newcommand{\E}{\alg{E}}
\newcommand{\cop}[1]{\mathbb{#1}}
\newcommand{\eq}{\approx}
\newcommand{\leqn}{\preccurlyeq}
\newcommand{\leqnv}{\preccurlyeq_{\V}}
\newcommand{\eleq}{\sqsubseteq}
\newcommand{\lgc}[1]{\mathrm{#1}}

\newcommand{\K}{{\ensuremath{\mathcal{K}}}}
\newcommand{\V}{{\ensuremath{\mathcal{V}}}}
\newcommand{\U}{{\ensuremath{\mathcal{U}}}}
\newcommand{\var}{{\ensuremath{\mathrm{Var}}}}
\newcommand{\FpV}{\Fp_{\V}}
\newcommand{\Fp}{\alg{Fp}}
\newcommand{\tp}{\mathrm{type}}
\newcommand{\cg}{\mathrm{Con}}
\newcommand{\ExUn}{\mathsf{E}}
\newcommand{\EUAlg}{\mathsf{C}}
\renewcommand{\U}[2]{\mathsf{U}_{#1}(#2)}

\newcommand{\FML}{\alg{Fm}_{\lang}(\omega)}
\newcommand{\FreeV}{\alg{F}_{\V}}
\newcommand{\Fml}{\alg{Fm}_{\lang}}

\begin{document}

\title[Exact Unification and Admissibility]{Exact Unification and Admissibility}

\author[L.~M.~Cabrer]{Leonardo M. Cabrer\rsuper a}	
\address{{\lsuper a}Department of Statistics, Computer Science and Applications, University of Florence, Italy}	
\thanks{{\lsuper a}The research of the first author was supported by a Marie Curie Intra European Fellowship within  the 
European Community's Seventh Framework Programme [FP7/2007-2013] under Grant Agreement n.~326202.}	
\email{l.cabrer@disia.unifi.it}  
\author[G.~Metcalfe]{George Metcalfe\rsuper b}	
\address{{\lsuper b}Mathematical Institute, University of Bern,  Switzerland}	
\email{george.metcalfe@math.unibe.ch}  
\thanks{{\lsuper b}The second author acknowledges support from Swiss National Science Foundation grant 200021{\_}146748.}



\begin{abstract}
A new hierarchy of ``exact" unification types is introduced, motivated by the study of admissible rules 
for equational classes and non-classical logics. 
In this setting, unifiers of identities in an equational class are preordered, not by instantiation, 
but rather by inclusion over the corresponding sets of unified identities. Minimal complete sets of 
unifiers under this new preordering always have a smaller or equal cardinality than those provided by 
the standard instantiation preordering, and in significant cases a dramatic reduction 
may be observed. In particular, the classes of distributive lattices, idempotent semigroups, and MV-algebras, 
which all have nullary unification type, have unitary or finitary exact type. 
These results are obtained via an algebraic interpretation of exact unification, inspired by Ghilardi's 
algebraic approach to equational unification. 
\end{abstract}


\maketitle

 
\section{Introduction}

It has long been recognized that the study of admissible rules is inextricably related to the theory 
of equational unification (see, e.g.,~\cite{Ryb97,Ghi99,Ghi00}). Indeed, from an algebraic perspective, admissibility 
of clauses in an equational class of algebras  may be understood as a generalization of unifiability of finite sets 
of identities in the class, and 
conversely,  checking admissibility may be reduced to comparing certain sets of unifiers. 
This paper provides a new classification of equational unification problems 
that simplifies these reductions for certain classes, including  distributive lattices, idempotent semigroups, and 
MV-algebras.

Let us fix an equational class of algebras $\V$ for a finite algebraic language 
$\lang$.\footnote{The reader is referred to~\cite{BS81} for basic concepts and results 
from universal algebra.} We denote  the {\em formula algebra} of $\lang$ 
over a set of variables $X$ by $\Fml(X)$ and write $\var(\Si)$ to denote the set of variables occurring in 
a set of $\lang$-identities $\Si$.  A substitution (homomorphism) $\si \colon \Fml(X) \to\FML$ is called a 
{\em $\V$-unifier (over $X$)} of a set of $\lang$-identities $\Si$ with $\var(\Si) \subseteq X$ if  for all $\f \eq \p$ in $\Si$,
\[
\V \models \si(\f) \eq \si(\p).
\] 
In this case, $\Si$ is also said to be {\em $\V$-unifiable}.

A {\em clause} $\Si \imp\De$, defined as an ordered pair $(\Si,\De)$ of finite sets of $\lang$-identities, 
 is called {\em $\V$-admissible} if for each substitution $\si \colon \Fml(\var(\Si \cup \De)) \to\FML$, 
 \[
 \mbox{$\si$ is a $\V$-unifier of $\Si$} 
 \quad \Longrightarrow \quad
  \mbox{$\si$ is a $\V$-unifier of some member of $\De$.} 
 \]
 In particular, $\Si$ is $\V$-unifiable if and only if (henceforth iff) $\Si \imp \emptyset$ is not $\V$-admissible.

Suppose now that the unification type of $\V$ is at most finitary: that is, every $\V$-unifier 
of a set of $\lang$-identities $\Si$ over a finite set $X \supseteq \var(\Si)$  is equivalent in $\V$ to one 
of a finite set $S$ of $\V$-unifiers of $\Si$ over $X$ composed with a further substitution. 
Then any clause $\Si \imp \De$ satisfying $\var(\De) \subseteq X$ 
is $\V$-admissible iff each member  of $S$ is a $\V$-unifier of a member of $\De$. 
If there is an algorithm for determining such a finite basis set $S$ for $\Si$ and the equational theory of $\V$ is decidable, 
then checking $\V$-admissibility is also decidable. This observation, together with the pioneering work of Ghilardi on 
equational unification for classes of Heyting and modal algebras~\cite{Ghi99,Ghi00}, has led to a wealth of decidability, 
complexity, and axiomatization results for admissibility in these classes and corresponding intermediate and modal  
logics~\cite{Iem01,Iem05,Jer05,CM10,BR11a,BR11b,OR13,GI14}.

The success of this approach to admissibility appears to rely on considering varieties with at most finitary unification type. 
That this is not a necessary condition, however,  is illustrated by the case of MV-algebras, the algebraic semantics of \L ukasiewicz 
infinite-valued logic (see~\cite{COM99} for details). Decidability, complexity, and axiomatization results for admissibility in 
MV-algebras have been established by Je{\v r}{\'a}bek~\cite{Jer09a,Jer09b,Jer13b} 
via a similar reduction of finite sets of identities to finite approximating sets of identities.  
On the other hand, it has been shown by Marra and Spada that 
the class of MV-algebras has nullary unification type~\cite{MS13}. This means  that there are finite sets of identities 
for which no finite basis of unifiers exists. Further examples of this discrepancy may be found in~\cite{CM13,MR13}, 
including  the  simple example of the class of distributive lattices where admissibility and validity of clauses coincide 
but unification is nullary. 

As mentioned above, it is possible to check the $\V$-admissibility of a clause $\Si \imp \De$ by checking 
that every $\V$-unifier of $\Si$ in a certain basis set $\V$-unifies some member of $\De$. Such a basis set $S$ 
typically has the property that every other $\V$-unifier of $\Si$ is obtained, modulo equivalence in $\V$, 
by applying a further substitution to a 
member of $S$. The starting point for this paper is the observation that a weaker condition on $S$ suffices, 
leading potentially to smaller basis sets of $\V$-unifiers. For checking $\V$-admissibility, it is enough  
that any $\V$-unifier of $\Si$ over  a finite set  $X \supseteq \var(\Si)$ is a $\V$-unifier of all identities  with 
variables in $X$ that are $\V$-unified by some particular member of $S$. Then  $\Si \imp \De$ with 
$\var(\De) \subseteq X$ is $\V$-admissible iff each member of $S$ is a $\V$-unifier of a member of $\De$. 
This observation leads to a new preordering of $\V$-unifiers and hierarchy of ``exact'' unification types. 

We also provide here an algebraic characterization of exact unification, where finite sets of identities are 
represented by finitely presented algebras. In Ghilardi's algebraic account of (standard) 
unification, unifiers are homomorphisms 
from finitely presented algebras into projective algebras of the class, preordered by composition of  
homomorphisms~\cite{Ghi97}. Coexact unifiers are defined here as homomorphisms from finitely presented algebras onto 
 algebras that embed into the $\omega$-generated free algebra of the class; 
the preordering remains the same. This contrasts with the syntactic account of exact unification where the 
unifiers are unchanged but a new preorder is introduced. Nevertheless, the syntactic and algebraic exact 
unification types coincide as in the standard unification setting.

Although certain equational classes have the same exact type as unification type (e.g., any equational class of 
unitary unification type will have unitary exact type),  we also obtain examples where the exact type is strictly smaller.  
In particular, the classes of distributive lattices and Stone algebras have nullary unification type but unitary exact type, while the classes of 
idempotent semigroups, pseudo-complemented distributive lattices, Kleene algebras, De Morgan algebras, 
and MV-algebras all have nullary unification type but finitary exact type. 
We also provide an example (due to R.~Willard) of an equational class that has infinitary unification type but finitary exact type.\footnote{
Another alternative hierarchy of unification types is obtained by considering left and right substitutions and so-called 
essential unifiers~\cite{HS06}. Although some of the advantages of this hierarchy are shared by our approach (e.g., the type of 
idempotent semigroups is in both cases finitary, contrasting with the fact that the unification type is nullary), the preordering 
for essential unification is different to the preordering presented here and not suited to reasoning about admissibility.}

We proceed as follows. In Section~\ref{s:equational_unification_and_admissibility}, we recall some standard notions of equational 
unification and admissibility, and describe Ghilardi's algebraic account of unification types. In Section~\ref{s:exact}, 
we introduce the new exact unification preordering and exact  types, providing also an algebraic interpretation and 
applications. Several cases studies are considered in Section~\ref{s:case_studies} 
and some ideas for further research are presented in Section~\ref{s:concluding_remarks}.


\section{Equational Unification and Admissibility}\label{s:equational_unification_and_admissibility}

In this section, we describe briefly some key notions from the theory of 
equational unification (referring to~\cite{BS01} for further details) 
and their relevance to the study of admissible rules. In particular, we recall 
the unification type of a finite set of identities in an equational class and 
the algebraic account of unification provided by Ghilardi in~\cite{Ghi97}. 
These notions and also those to appear in subsequent sections are most elegantly 
presented in the general setting of preordered sets.

Let $\pos{P} = \langle P, \le \rangle$ be a preordered set (i.e., $\le$ is a reflexive and transitive binary relation on $P$). 
A {\it complete} set for $\pos{P}$ is a subset $M \subseteq P$ such that for every $x \in P$, there exists $y \in M$ satisfying 
$x \le y$. A complete set $M$ for $\pos{P}$ is called a {\it $\mu$-set} for $\pos{P}$ if  $x \not \le y$ and $y \not \le x$ 
for all distinct $x,y \in M$. It is easily seen that if $\pos{P}$ has a $\mu$-set, then every $\mu$-set of $\pos{P}$ has the 
same cardinality. Hence $\pos{P}$ may be said to be {\em nullary} if it  has no $\mu$-sets ($\tp(\pos{P}) = 0$), {\em infinitary} if it has 
a $\mu$-set of infinite cardinality ($\tp(\pos{P}) = \infty$), {\em finitary} if it has a finite $\mu$-set of cardinality greater than $1$ 
($\tp(\pos{P}) = \omega$), and {\em unitary} if it has a $\mu$-set of cardinality~1 ($\tp(\pos{P}) = 1$). These types are ordered as 
follows: $1<\omega<\infty<0$. 

The following trivial but helpful observation confirms that the type of a preordered set depends only on 
its corresponding quotient poset.

\begin{lem}[{\cite[Lemma~2.1]{Ghi97}}]\label{Lemma:EquivPreorder}
Suppose that two preordered sets $\langle P, \le \rangle$ and  $\langle Q, \le \rangle$ 
are equivalent: i.e., there exists a map $e \colon P \to Q$ such that
\begin{enumerate}
\item	for each $q \in Q$, there is a $p \in P$ such that $e(p) \le q$ and $q \le e(p)$, and\smallskip
\item]	for all $p_1,p_2 \in P$, $p_1 \le p_2$ iff $e(p_1) \le e(p_2)$.
\end{enumerate}
Then $\langle P, \le \rangle$ and  $\langle Q, \le \rangle$ have the same type.
\end{lem}

We turn our attention now to the syntactic account of equational unification. 
Let us fix a finite algebraic language $\lang$ and an equational class  $\V$ of $\lang$-algebras 
(equivalently, a variety: a class of $\lang$-algebras closed under taking products, subalgebras, and homomorphic 
images).\footnote{The results of this paper also hold for quasi-equational classes and, more generally, 
any class of algebras that contains finitely presented algebras for all finite presentations (equivalently, 
{\em prevarieties}: classes closed under taking products, subalgebras, and isomorphic images~\cite{Gor98}). 
However, as  the vast majority of cases considered in 
the literature are equational classes, we restrict our account to this slightly simpler setting.} 
Consider  a finite set  $X \subseteq \omega$ and substitutions $\si_i \colon \Fml(X) \to\FML$ for $i = 1,2$. 
We say that $\si_1$ is {\em more general than $\si_2$ in $\V$}, written 
$\si_2 \leqnv \si_1$, if there exists a substitution $\tau \colon \FML\to \FML$ such that 
$\V \models \tau(\si_1(x)) \eq \si_2(x)$ for all $x \in X$.

Let $\Si$ be a finite set of $\lang$-identities 
and  let $X\supseteq \var(\Si)$ be a finite set of variables. Then $\U{\V}{\Si,X}$ is defined as the set of 
$\V$-unifiers of $\Si$ over $X$ preordered by $\leqnv$, and we let $\U{\V}{\Si} = \U{\V}{\Si,\var(\Si)}$. 
Note also that, trivially, $\U{\V}{\Si,X} = \U{\V}{\Si \cup \{x \eq x \mid x \in X\}}$.

For $\U{\V}{\Si} \not = \emptyset$,  the {\em $\V$-unification type} of 
$\Si$ is defined as  $\tp(\U{\V}{\Si})$.  The {\em unification type of $\V$} is then the maximal type of a $\V$-unifiable 
finite set $\Si$ of $\lang$-identities.

\begin{exa}\label{e:equnif}
Equational unification has been studied for a wide range of equational classes. 
For {\em syntactic unification}, where $\V$ is the class of all $\lang$-algebras, 
every unifiable finite set $\Si$ of $\lang$-identities has  a most general unifier; 
that is, syntactic unification is unitary (see, e.g.,~\cite{BS01}). The class of {\em Boolean algebras} 
is also unitary~\cite{BS87}: 
if $\{\f \eq \p\}$  has a unifier $\si_0$, then $\si(x) = (\lnot (\f + \p) \land x) \lor ((\f + \p) \land \si_0(x))$  for each variable $x$ 
(where $+$ is the symmetric difference operation)  defines a most general unifier. 
The class of {\em Heyting algebras} is not unitary; e.g., 
$\{x \lor y \eq \top\}$ has a $\mu$-set of unifiers $\{\si_1,\si_2\}$ where $\si_1(x) = \top$, $\si_1(y) = y$, 
$\si_2(x) = x$, $\si_2(y) = \top$. However,  this class is finitary~\cite{Ghi99}. More problematically, 
the class of {\em semigroups} is infinitary~\cite{Plo72}; e.g., $\{x \cdot y \eq y \cdot x\}$ has  a $\mu$-set 
$\{\si_{m,n} \mid \gcd(m,n)= 1\}$ where $\si_{m,n}(x) = z^m$ and $\si_{m,n}(y) = z^n$. 
Moreover, many familiar classes of algebras are nullary; e.g., in  
the class of {\em distributive lattices} (see~\cite{Ghi99}), 
$\{x \land y \eq z \lor w\}$ has no $\mu$-set. Other nullary classes of algebras 
include {\em idempotent semigroups}~\cite{Ba86}, {\em pseudo-complemented 
distributive lattices}~\cite{Ghi97},  {\em MV-algebras}~\cite{MS13}, and {\em modal algebras} 
for the logic $\lgc{K}$~\cite{Jer13}.
\end{exa}

Let us now recall Ghilardi's algebraic account of equational unification~\cite{Ghi97}. 
We denote  the free $\lang$-algebra of $\V$ over a set of 
variables $X$ by $\FreeV(X)$ and let $h_\V \colon \Fml(X) \to \FreeV(X)$ be the canonical homomorphism 
acting as the identity on $X$. 
Given a finite set of $\lang$-identities $\Si$ and a finite set of variables $X \supseteq \var(\Si)$, we denote by $\FpV(\Si,X)$, 
the algebra in $\V$ finitely presented by $\Si$ and $X$: that is, the quotient algebra $\FreeV(X) / \Theta_\Si$ where 
$\Theta_\Si$ is the congruence on $\FreeV(X)$ generated by the set $\{(h_\V(\f),h_\V(\p)) \mid \f \eq \p \in \Si\}$. We also let
$\mathsf{FP}(\V)$ denote the class of finitely presented algebras of $\V$.

Given $\A \in \mathsf{FP}(\V)$, a homomorphism $u \colon \A \to \B$ is called a {\em unifier} of $\A$ if 
 $\B \in \mathsf{FP}(\V)$ is {\em projective} in $\V$: that is, there exist homomorphisms $\iota \colon \B \to \FreeV(\omega)$ and 
 $\rho \colon \FreeV(\omega) \to \B$ such that $\rho \circ \iota$ is the identity map on $B$. Let $u_i \colon \A \to \B_i$ for $i=1,2$ be unifiers 
for $\A$. Then $u_1$ is {\em more general than} $u_2$, written $u_2 \le u_1$, if there exists a homomorphism 
$f \colon \B_1 \to \B_2$ such that $f \circ u_1 = u_2$.

Let $\U{\V}{\A}$ be the set of unifiers of $\A \in \mathsf{FP}(\V)$ preordered by $\le$.  
For $\U{\V}{\A}\neq\emptyset$, the {\em unification type of $\A$ in $\V$} is 
defined as $\tp(\U{\V}{\A})$ and the {\em algebraic unification type} of $\V$ is the maximal type of 
$\A$ in $\mathsf{FP}(\V)$ such that $\U{\V}{\A}\neq\emptyset$.

\begin{thm}[{\cite[Theorem~4.1]{Ghi97}}]
Let $\V$ be an equational class and let $\Si$ be a finite $\V$-unifiable set of $\lang$-identities. 
Then for any finite set of variables $X\supseteq \var(\Si)$:
\[
\tp\bigl(\U{\V}{\Si,X}\bigr) = \tp\bigl(\U{\V}{\FpV(\Si,X)}\bigr).
\]
Hence the algebraic unification type of $\V$ coincides with the unification type of $\V$. 
\end{thm}

Let us see now how these ideas relate to the notion of admissibility defined in the introduction. 
Recall that the {\em kernel} of a homomorphism $h \colon \A \to \B$  is defined as 
\[
\ker(h) = \{(a,b) \in A^2 \mid h(a)=h(b)\}.
\]
In what follows, we will freely identify $\lang$-identities with pairs of $\lang$-formulas. We will also say that an  
$\lang$-clause $\Si \imp \De$ is {\em valid} in a class of $\lang$-algebras $\K$, written $\K \models \Si \imp \De$, 
if the universal sentence $(\forall \bar{x})(\bigwedge \Si \imp \bigvee \De)$ is valid in each algebra in $\K$. 
The admissibility of an $\lang$-clause can then be reformulated as follows:

\begin{lem}\label{Lemma:Crucial}
Let $\Si\cup\De$ be a finite set of $\lang$-identities with $\var(\Si \cup \De) \subseteq X$. 
Then the following  are equivalent:
\begin{enumerate}[label=(\roman*)]
\item		$\Si\imp \De$ is $\V$-admissible.
\item 	For each substitution $\si\colon \Fml(X) \to \FML$ such that $\Si\subseteq \ker(h_{\V}\circ\si)$, 
			\[
			\De\cap\ker(h_{\V}\circ\si)\neq\emptyset.
			\]
\item	$\FreeV(\omega)\models \Si\imp \De$. 
\end{enumerate}
If in particular $\De=\{ \f\eq\p\}$, then (i)-(iii) above are also equivalent to
\begin{enumerate}[label=(\roman*),start=4]
\item[(iv)] 	$(\f,\p)\in\bigcap \{\ker(h_{\V}\circ \si)\mid \si\colon \Fml(X) \to \FML\mbox{ and }\Si\subseteq \ker(h_{\V}\circ\si)\}$.
\end{enumerate}
\end{lem}
\begin{proof}
We give a proof here of this well known equivalence 
(see, e.g.,~\cite{Ryb97,MR13})  for the sake of completeness.

(i)$\Leftrightarrow$(ii) Recall (see~\cite[Corollary II.11.6]{BS81}) that, for each 
$\lang$-identity $\f \eq \p$:
\[
\V \models \f \eq \p \qquad \Longleftrightarrow \qquad \FreeV(\omega) \models \f \eq \p  \qquad \Longleftrightarrow 
\qquad h_\V(\f) = h_\V(\p). 
\]
Hence a  substitution $\si\colon \Fml(X) \to \FML$ satisfies $\Si\subseteq \ker(h_{\V}\circ\si)$ (i.e., 
$h_{\V}(\si(\f)) = h_{\V}(\si(\p))$ for all $\f \eq \p \in \Si$) iff $\V \models \si(\f) \eq \si(\p)$  for all $\f \eq \p \in \Si$, 
that is, iff $\si$ is a $\V$-unifier of $\Si$. Similarly, $\De\cap\ker(h_{\V}\circ\si)\neq\emptyset$ iff $\si$ is a
$\V$-unifier of some member of $\De$. So (ii) holds iff $\Si\imp \De$ is $\V$-admissible.

(i)$\Rightarrow$(iii) 
Suppose that $\Si \imp \De$ is $\V$-admissible and let $g \colon \FML \rightarrow \FreeV(\omega)$ be a 
homomorphism such that $\Si \subseteq \ker g$. Let $\si$ be a map 
sending each variable $x$ to a member of the equivalence class $g(x)$. By the universal mapping property for 
$\FML$, this extends to a homomorphism $\si \colon \FML \rightarrow \FML$. 
But $h_\V (\si(x)) = g(x)$ for each variable $x$, so $h_\V \circ \si = g$. 
Hence, for each  $\f' \eq \p' \in \Si$, it holds that $h_\V(\si(\f')) =  h_\V(\si(\p'))$ and therefore $\V \models\si(\f') \eq \si(\p')$. 
So $\si$ is a $\V$-unifier of $\Si$ and, by assumption, $\V \models \si(\f) \eq \si(\p)$ for some $\f\eq \p\in\De$. 
It follows that  
$g(\f) = h_\V(\si(\f)) = h_\V(\si(\p)) = g(\p)$ as required.

(iii)$\Rightarrow$(ii)
 Consider a substitution $\si\colon \Fml(X) \to \FML$ such that  
 $\Si\subseteq \ker(h_{\V}\circ\si)$; that is, $\V \models \si(\f) \eq \si(\p)$ for all $\f \eq \p \in \Si$. 
So $\FreeV(\omega)\models \si(\f) \eq \si(\p)$ for all $\f \eq \p \in \Si$.
 By assumption, there exists $\f\eq\p\in\De$ such that  $\FreeV(\omega)\models \si(\f)\eq\si(\p)$. But then 
 also  $\V \models \si(\f)\eq\si(\p)$ and, as required, $(\f,\p)\in\ker(h_{\V}\circ\si)\cap \De$.

If $\De=\{\f\eq\p\}$, then (ii) is clearly equivalent to (iv).
\end{proof}

Suppose now that $\V$ is any equational class of $\lang$-algebras and that $\Si$ and $\De$ are finite sets 
of $\lang$-identities. Given any complete set $S$
for $\U{\V}{\Si,\var(\Si \cup \De)}$,  it follows directly that
\[
\Si \imp \De \mbox{ is $\V$-admissible}  \quad \Longleftrightarrow \quad 
\mbox{each $\si \in S$ is a $\V$-unifier of some $\f \eq \p \in \De$.}
\]
Moreover, if $\V$ is unitary or finitary and there exists an algorithm for finding finite complete sets of unifiers, 
then checking admissibility in $\V$ is decidable whenever the equational theory of $\V$ is decidable. 
There are, however, important equational classes having infinitary or nullary unification type for which 
such a method is unavailable. The starting point for the new approach described below is the observation 
that the above equivalence can hold even for a set $S$ that is not complete for the $\leqn_\V$-preordered 
set of $\V$-unifiers. It suffices rather that each $\si \in \U{\V}{\Si,\var(\Si \cup \De)}$ is a $\V$-unifier 
of all the identities $\V$-unified by some particular member of~$S$.


\section{Exact Unification}\label{s:exact}

We begin by defining a new preorder on substitutions relative to a fixed equational class of $\lang$-algebras $\V$. 
Let  $X$ be a finite set of variables and let $\si_i \colon \Fml(X)\to\FML$ 
be substitutions for $i=1,2$.  We write $\si_2\eleq_{\V} \si_1$ if all identities $\V$-unified by $\si_1$ are $\V$-unified by $\si_2$. 
More precisely:
\[
\si_2 \eleq_{\V} \si_1 \quad \Longleftrightarrow \quad  \ker(h_{\V}\circ\si_1)\subseteq\ker(h_{\V}\circ\si_2).
\]
Clearly, $\eleq_\V$ is a preorder on substitutions of the form $\si \colon \Fml(X) \to\FML$. 

\begin{lem}\label{l:easier}
For any finite set $X$ and substitutions $\si_i \colon \Fml(X)\to\FML$  for $i=1,2$:
\[
\si_2\leqnv \si_1 \quad \Longrightarrow \quad \si_2\eleq_{\V} \si_1.
\]
Moreover, if  $h_\V \circ \si_1 \circ \si_1 = h_\V \circ \si_1$, then
\[
\si_2\leqnv \si_1 \quad \Longleftrightarrow \quad \si_2\eleq_{\V} \si_1.
\]
\end{lem}
\begin{proof}
Suppose that $\si_2\leqnv \si_1$. Then there exists a substitution  $\tau \colon \FML\to\FML$ such that 
$h_\V \circ \tau \circ \si_1 = h_\V \circ \si_2$. Consider $(\f,\p) \in \ker(h_{\V}\circ\si_1)$; i.e., 
$h_\V(\si_1(\f)) = h_\V (\si_1(\p))$. Then, since $\V \models \si_1(\f) \eq \si_1(\p)$ implies 
$\V \models \tau(\si_1(\f)) \eq \tau(\si_1(\p))$, also
\[
h_\V (\si_2(\f)) = h_\V (\tau ( \si_1(\f))) = h_\V( \tau ( \si_1(\p))) = h_\V  ( \si_2(\p)). 
\]
That is, $(\f,\p) \in\ker(h_{\V}\circ\si_2)$. So $\si_2\eleq_{\V} \si_1$.

Now suppose that $h_\V \circ \si_1 \circ \si_1 = h_\V \circ \si_1$ and $\si_2\eleq_{\V} \si_1$.
Then for each  $x\in X$, $h_\V ( \si_1 ( \si_1(x))) = h_\V ( \si_1(x))$. Hence  
$(\si_1(x), x) \in\ker(h_{\V}\circ\si_1)\subseteq \ker(h_{\V}\circ\si_2)$. That is, 
$h_\V(\si_2(\si_1(x))) = h_\V(\si_2(x))$. 
It follows that $h_\V \circ \si_2 \circ \si_1 = h_\V \circ \si_2$. So $\si_2\leqnv \si_1$. 
\end{proof}

Now let  $\Si$ be a  finite set  of $\lang$-identities and let $X\supseteq \var(\Si)$ be a finite set 
of variables.  $\ExUn_{\V}(\Si,X)$ is defined as the set of $\V$-unifiers of $\Si$ over $X$ preordered 
by $\eleq_{\V}$, and we denote $\ExUn_{\V}(\Si,\var(\Si))$ by $\ExUn_{\V}(\Si)$. 
If $\ExUn_{\V}(\Si) \neq \emptyset$, then 
 the {\it exact type of $\Si$ in $\V$} is defined as $\tp(\ExUn_{\V}(\Si))$. 
We also define the {\em exact  type of $\V$} to be
 the maximal exact type of a $\V$-unifiable finite set $\Si$ of $\lang$-identities in $\V$.
 
Note that, because $\si_2\leqnv \si_1$ implies $\si_2\eleq_{\V} \si_1$ (Lemma~\ref{l:easier}), every 
complete set for $\U{\V}{\Si}$ is also a complete set for  $\ExUn_{\V}(\Si)$.  Hence, for $\tp(\U{\V}{\Si}) \in \{1,\omega\}$,
\[
\tp(\ExUn_{\V}(\Si)) \le \tp(\U{\V}{\Si}),
\] 
and if $\tp(\ExUn_{\V}(\Si))\in \{\infty,0\}$, then also $\tp(\U{\V}{\Si}) \in \{\infty,0\}$. 

The following relationship between exact unification and admissibility in $\V$ is an immediate consequence of 
Lemma~\ref{Lemma:Crucial}.

\begin{cor}\label{Cor:AdmissibleExact}
Let $\Si\cup\De$ be a finite set of $\lang$-identities and let $X \supseteq \var(\Si \cup \De)$ be a finite set of 
variables. If $S$ is  a complete set for $\ExUn_{\V}(\Si,X)$, then the following  are equivalent:
\begin{enumerate}[label=(\roman*)]
\item	$\Si\imp \De$ is $\V$-admissible.
\item 	Each $\si \in S$ is a $\V$-unifier of some $\f \eq \p \in \De$. 
\item 	For each $\si\in S$, $\De\cap\ker(h_{\V}\circ \si)\neq\emptyset$.\qed
\end{enumerate} 
\end{cor}

\noindent Note (again) that if $\V$ has unitary or finitary exact type and there exists an algorithm for finding finite complete 
sets of unifiers, then checking admissibility in $\V$ is decidable whenever the equational theory of $\V$ is decidable. 
We also observe  that the cardinality of a finite complete set of unifiers for the premises of a clause 
 provides a bound for the number of consequences relevant for determining its admissibility.

\begin{prop}\label{Prop:TypeAndSizeAdm}
If an $\lang$-clause $\Si\imp\De$ is $\V$-admissible and $S$ is a finite complete set for 
$\ExUn_{\V}(\Si,\var(\Si\cup\De))$ then there exists 
$\De'\subseteq \De$ such that 
$|\De'|\leq |S|$ and $\Si\imp\De'$ 
is $\V$-admissible.
\end{prop}
\begin{proof}
Let  $S=\{\si_1,\ldots,\si_n\}$ be a complete set for $\ExUn_{\V}(\Si,\var(\Si\cup\De))$. 
By Lemma~\ref{Lemma:Crucial}, for each $i\in\{1,\ldots,n\}$, there exists 
$\f_i\eq\p_i \in \De$ such that $(\f_i,\p_i)\in\ker(h_{\V}\circ \si_i)$. Let $\De'=\{\f_1\eq\p_1, \ldots, \f_n \eq \p_n\}$. 
But $S$ is a complete set for   $\ExUn_{\V}(\Si,\var(\Si\cup\De))$, so by Corollary~\ref{Cor:AdmissibleExact}, 
also $\Si\imp\De'$ is $\V$-admissible. 
\end{proof}

\begin{prop}\label{prop:admred}
Let $\Si$ be a finite set of $\lang$-identities and  let $X\supseteq \var(\Si)$ be a finite set of variables. 
If $\tp(\ExUn_{\V}(\Si,X))=1$, then the following condition holds:

\begin{itemize}[label=$(\star)$]
\item
Whenever $\Si\imp\De$ is $\V$-admissible with $\var(\De) \subseteq X$, there exists $\f\eq\p\in\De$ such that 
$\Si\imp\f\eq\p$ is $\V$-admissible. 
\end{itemize}
Conversely,  if  $\tp(\ExUn_{\V}(\Si,X))\in\{1,\omega\}$  and $\Si$ has property $(\star)$, then 
$\tp(\ExUn_{\V}(\Si,X))=1$.
\end{prop}

\begin{proof}
The first claim follows immediately from the previous proposition, noting that if $\tp(\ExUn_{\V}(\Si,X))=1$ and 
$\Si\imp\De$ is $\V$-admissible, then $\De \not = \emptyset$. 
For the second claim, 
assume that $\tp(\ExUn_{\V}(\Si,X))\in\{1,\omega\}$ and that  $\Si$ has property $(\star)$. 
Then there exists a $\mu$-set $\{\si_1,\ldots,\si_n\}$ for $\ExUn_{\V}(\Si,X)$. 
For each $i,j\in\{1,\ldots, n\}$ such that $i\neq j$, consider $(\f_{ij},\p_{ij})\in\ker(h_{\V}\circ\si_i)\setminus\ker(h_{\V}\circ\si_j)$. Let \[\De=\{\f_{ij}\eq\p_{ij}\mid i,j\in\{1,\ldots, n\} \mbox{ and }i\neq j\}.\]  
Suppose for a contradiction that $n\neq 1$ and hence $\De\neq\emptyset$.
As $\{\si_1,\ldots,\si_n\} $ is a $\mu$-set for $\ExUn_{\V}(\Si,X)$, it follows by Corollary~\ref{Cor:AdmissibleExact} that 
$\Si\imp\De$ is $\V$-admissible. But then, by assumption, there exists $\f_{ij}\eq\p_{ij}\in\De$ such that  
$\Si\imp\f_{ij}\eq\p_{ij}$ is $\V$-admissible, contradicting the fact that $\V\not\models\si_j(\f_{ij})\eq\si_j(\p_{ij})$. 
So $n=1$. Hence $\tp(\ExUn_{\V}(\Si,X))=1$.
\end{proof}

We turn our attention now to the algebraic interpretation of exact unification. Note that while the syntactic accounts of 
exact unification and standard unification use the same sets of unifiers but consider different preorders, 
the algebraic interpretations of exact unification and standard unification share the same preorder but differ in the sets of (algebraic) unifiers considered (see also the comments after Theorem~\ref{Theo:EqualTypes}).

Following~\cite{DeJongh82}, a finite set  of $\lang$-identities $\Si$ will be called {\em exact in $\V$} if there exists 
a substitution $\si\colon\Fml(\var(\Si))\rightarrow\FML$ such that for all $\f,\p\in\Fml(\var(\Si))$, 
\[
\V\models \Si\imp \f\eq\p 
\quad \Longleftrightarrow \quad 
\V\models \si(\f)\eq \si(\p).
\]
Note that, by definition, every finite set of $\lang$-identities that is exact in $\V$ is $\V$-unifiable.

Let $\Si$ be a  finite set of $\lang$-identities and let $X \supseteq \var(\Si)$ be a finite set of variables. 
We define $\rho_{(\Si,X,\V)}\colon\FreeV(X)\rightarrow \FpV(\Si,X)$ as the canonical quotient homomorphism 
from the free algebra $\FreeV(X)$ to the finitely presented algebra $\FpV(\Si,X)$.

\begin{lem}\label{Lem:ExactAsSubAlg}
A finite set $\Si$ of $\lang$-identities is exact in $\V$ iff 
$\FpV\bigl(\Si,\var(\Si)\bigr)\in\cop{IS}\bigl(\FreeV(\omega)\bigr)$. 
\end{lem}
\begin{proof}
($\Rightarrow$) Let $X=\var(\Si)$ and let $\si\colon\Fml(X)\rightarrow\FML$  be a substitution 
such that for all $\f,\p\in\Fml(X)$, $\V\models \Si\imp \f\eq\p$
iff $\V\models \si(\f)\eq \si(\p)$. That is, $\V\models \Si\imp \f\eq\p$ iff $h_{\V}(\si(\f))=h_{\V}(\si(\p))$. 
Let $g\colon \FreeV(X)\rightarrow \FreeV(\omega)$ be the unique homomorphism
satisfying $h_{\V}\circ\si=g\circ h_{\V}$. Then $h_{\V}(\Si)\subseteq \ker(g)$ and we obtain 
a unique homomorphism  $s\colon \FpV(\Si,X)\rightarrow\FreeV(\omega)$  such that 
$s\circ\rho_{(\Si,X,\V)}=g$. Consider $a,b\in \FpV(\Si,X)$ such that $s(a)=s(b)$ and 
$\f,\p\in\Fml(X)$ satisfying  $\rho_{(\Si,X,\V)}(h_{\V}(\f))=a$ and $\rho_{(\Si,X,\V)}(h_{\V}(\p))=b$. 
It follows that
\[
h_{\V}(\si(\f))=g( h_{\V}(\f))=s(\rho_{(\Si,X,\V)}( h_{\V}(\f)))=s(a)=s(b) 
= h_{\V}(\si(\p)).
\]
So, by assumption, $\V \models \Si \imp \f\eq\p$. But then, since  
$\rho_{(\Si,X,\V)}(h_{\V}(\f'))=\rho_{(\Si,X,\V)}(h_{\V}(\p'))$ for 
all $\f' \eq \p' \in \Si$, it follows that
\[
a=\rho_{(\Si,X,\V)}(h_{\V}(\f))=\rho_{(\Si,X,\V)}(h_{\V}(\p))=b.
\]
Hence $s$ is a one-to-one homomorphism and $\FpV(\Si,X)\in\cop{IS}(\FreeV(\omega))$.

($\Leftarrow$) Let $X=\var(\Si)$ and let 
$s\colon \FpV(\Si,X)\rightarrow\FreeV(\omega)$ be a  one-to-one homomorphism. 
Consider a homomorphism $\si\colon\Fml(X)\rightarrow\FML$ satisfying
\[
\si(x)=\f_x\ \mbox{ for each } x\in X,
\]
where $\f_x$ is any formula such that $s(\rho_{(\Si,X,\V)}(x))=h_{\V}(\f_x)$. 
By induction on formula complexity, $s(\rho_{(\Si,X,\V)}(h_{\V}(\f))) = h_{\V}(\si(\f))$ for all 
$\f\in\Fml(X)$. But then for $\f,\p\in\Fml(X)$, using the fact that $s$ is one-to-one:
\[
\begin{array}{rcl}
\V\models \si(\f)\eq \si(\p) 	& \Longleftrightarrow & h_{\V}(\si(\f))=h_{\V}(\si(\p))\\[.05in]
								& \Longleftrightarrow & s(\rho_{(\Si,X,\V)}(h_{\V}(\f)))=s(\rho_{(\Si,X,\V)}(h_{\V}(\p)))\\[.05in]
								& \Longleftrightarrow & \rho_{(\Si,X,\V)}(h_{\V}(\f))=\rho_{(\Si,X,\V)}(h_{\V}(\p))\\[.05in]
								& \Longleftrightarrow & \V\models \Si\imp \f\eq\p. 
\end{array}
\]
That is, $\Si$  is exact in $\V$.
\end{proof}

An algebra~$\A \in \V$ is called {\em exact} in $\V$ if it is isomorphic to a finitely generated subalgebra of 
$\FreeV(\omega)$. By Lemma~\ref{Lem:ExactAsSubAlg} (see also~\cite{DeJongh82}), a finite  set of $\lang$-identities 
$\Si$ is exact in $\V$ iff the finitely presented algebra $\FpV(\Si,\var(\Si))$ is exact in $\V$.

Given $\A \in \mathsf{FP}(\V)$, an onto homomorphism 
$u \colon \A \to \E$ is called a {\em coexact unifier} for~$\A$ in $\V$ if  $\E$ is exact in $\V$. 
 Coexact unifiers are preordered in the same way as algebraic unifiers; that is, if $u_i \colon \A \to \E_i$ for $i=1,2$ are coexact unifiers for 
 $\A$ in $\V$, then $u_1 \le u_2$ iff there exists a homomorphism $f \colon \E_1 \to \E_2$ such that $f \circ u_1 = u_2$. 

 Let $\EUAlg_{\V}(\A)$ be the  set of coexact unifiers for $\AÊ\in \mathsf{FP}(\V)$ preordered by $\leq$.  
 If $\EUAlg_{\V}(\A)\neq\emptyset$,  then the 
 {\em exact type} of $\A$ is defined as the type of $\EUAlg_{\V}(\A)$. The {\em algebraic exact type} 
 of $\V$ is the maximal exact type of $\A$ in $\V$ such that $\EUAlg_{\V}(\A) \neq\emptyset$.

\begin{thm}\label{Theo:EqualTypes}
Let $\V$ be an equational class and let $\Si$ be a finite $\V$-unifiable set of $\lang$-identities. Then for any 
finite set of variables $X\supseteq \var(\Si)$:
\[
\tp\bigl(\ExUn_{\V}(\Si,X)\bigr)=\tp\bigl(\EUAlg_{\V}(\FpV(\Si,X))\bigr).
\]
Hence the exact type and the algebraic exact type of $\V$ coincide. 
\end{thm}

\begin{proof}
Consider
 $\si\colon\Fml(X)\to \FML$ 
in $\ExUn_{\V}(\Si,X)$. 
Let $\hat{\si}\colon\FreeV(X)\rightarrow h_{\V}(\si(\Fml(X)))$ be the unique homomorphism satisfying
\[
\hat{\si}(x)=
h_{\V}(\si(x)) \ \mbox{ for each }x\in X.
\]
Then $\Si\subseteq\ker(\hat{\si}\circ h_{\V})$, and there exists a homomorphism 
$u_{\si}\colon \FpV(\Si,X) \to h_{\V}(\si(\Fml(X)))$
such that 
\begin{equation}\label{Eq:SyntacticAlgebraicSquare}
u_{\si}\circ \rho_{(\Si,X,\V)}\circ h_{\V} = h_{\V}\circ\si.
\end{equation} 
Note that the map $u_{\si}$ is onto $h_{\V}(\si(\Fml(X)))$. Because $h_{\V}(\si(\Fml(X)))$ is a finitely generated 
subalgebra of $\FreeV(\omega)$, 
also $u_{\si}\in \EUAlg_{\V}(\FpV(\Si,X))$. 
It suffices now, by Lemma~\ref{Lemma:EquivPreorder}, to show that 
the assignment $\si \mapsto u_{\si}$ determines an equivalence between the preordered sets $\ExUn_{\V}(\Si,X)$ and 
$\EUAlg_{\V}(\FpV(\Si,X))$. 

(1) Let $u \colon\FpV(\Si,X)\to \E$ be a coexact unifier for $\FpV(\Si,X)$ in $\V$. 
Because $\E$ is exact in $\V$, there exists a one-to-one homomorphism 
$\iota\colon \E\to \FreeV(\omega)$. 
For each $x\in X$, consider $\f_x\in\Fml(\omega)$ such that 
$h_{\V}(\f_x)=\iota(u(\rho_{(\Si,X,\V)}(x)))$.
Let $\si\colon \Fml(X)\to \Fml(\omega)$ be the substitution defined
by $\si(x)=\f_x$ for each $x\in X$. 
It is straightforward to check that 
$\iota\circ u=u_{\si}$ and $\iota(\E)= u_{\si}(\FpV(\Si,X))$. Because $\iota$ is one-to-one, there exists a homomorphism 
$\eta\colon u_{\si}(\FpV(\Si,X))\to \E$  that is the inverse of $\iota$. Therefore $u$ and $u_{\si}$ 
are equivalent in the preorder $\EUAlg_{\V}(\FpV(\Si,X))$, i.e., $u \le u_{\si}$ and $u_{\si} \le u$.

(2) Using \eqref{Eq:SyntacticAlgebraicSquare}, for all $\si_1,\si_2\in \ExUn_{\V}(\Si,X)$:
\begin{eqnarray*}
\si_2\eleq_{\V} \si_1&\Longleftrightarrow& \ker(h_{\V}\circ\si_1)\subseteq\ker(h_{\V}\circ\si_2)\\
&\Longleftrightarrow&
\ker(u_{\si_1}\circ \rho_{(\Si,X,\V)})\subseteq\ker(u_{\si_2}\circ \rho_{(\Si,X,\V)})\\
&\Longleftrightarrow&\ker(u_{\si_1})\subseteq \ker(u_{\si_2}).
\end{eqnarray*}
Let us denote the codomains of $u_{\si_1}$ and $u_{\si_2}$ by $\E_1$ and $\E_2$, respectively.
Because $u_{\si_1}$ is onto~$\E_1$, also $\ker(u_{\si_1})\subseteq \ker(u_{\si_2})$ iff there exists 
$h\colon \E_1\to \E_2$ such that $h\circ u_{\si_1}=u_{\si_2}$, that is, $u_{\si_2}\leq u_{\si_1}$.
\end{proof}

In passing from Ghilardi's algebraic account of unification to algebraic 
coexact unifiers, we have modified the definition of unifiers but preserved 
the preorder. An alternative approach, perhaps closer to the syntactic approach to exact unification, would be to preserve the unifiers as maps from 
a finitely presented algebra into a projective algebra, modifying the preorder. 
However, the characterization provided here highlights the connection 
between coexact unifiers and certain congruences of the relevant finitely presented algebra.

Given an algebra  $\A$ in  $\V$, recall that $\cg(\A)$ denotes the set of congruences on $\A$. 
We  let $\cg_e(\A)$ denote the set of congruences $\theta \in \cg(\A)$ of $\A$ such that the quotient $\A/\theta$ is exact in $\V$; i.e., 
\[
\cg_e(\A)=\{\theta\in\cg(\A)\mid \A/\theta\in\cop{IS}(\FreeV(\omega))\}.
\]

\begin{thm} \label{t:kerequiv}
Let $\V$ be an equational class and $\A \in \mathsf{FP}(\V)$.

\begin{enumerate}[label=\({\alph*}]

\item	For any onto homomorphism $u\colon \A\to \B$:			
			\[
			u\in \EUAlg_{\V}(\A) \quad \Longleftrightarrow \quad \ker(u)\in \cg_e(\A).
			\]
\item	For all $u,v\in \EUAlg_{\V}(\A)$:
			\[
			u\leq v\quad \Longleftrightarrow \quad\ker(v)\subseteq \ker(u).
			\]
\end{enumerate}
\noindent
Hence $\ker\colon \EUAlg_{\V}(\A)\to \cg_e(\A)$ determines an equivalence  (i.e., $\ker$ satisfies  \emph{(1)} and \emph{(2)} of 
Lemma~\emph{\ref{Lemma:EquivPreorder})} between the preordered sets $(\EUAlg_{\V}(\A),\le)$ 
and $(\cg_e(\A),\supseteq)$, and
\[
\tp\bigl(\EUAlg_{\V}(\A)\bigr) = \tp\bigl(\cg_e(\A)\bigr).
\]
\end{thm}
\begin{proof}
For (a), observe that, by the homomorphism theorem:
\[
u \in \EUAlg_{\V}(\A)
\quad \Longleftrightarrow \quad
u(\A)\in\cop{IS}(\FreeV(\omega))
\quad \Longleftrightarrow \quad
\ker(u)\in \cg_e(\A).
\]
For (b), observe that $u\leq v$ iff there exists a homomorphism $f\colon v(\A)\to u(\A)$ such that $f\circ v=u$ iff 
(as $v$ is onto) $\ker(v)\subseteq \ker(u)$.
\end{proof}

\begin{cor}\label{Cor:LocFin}
Let $\V$ be a locally finite equational class. Then $\tp(\EUAlg_{\V}(\A)) \in \{1,\omega\}$
for  each $\A \in \mathsf{FP}(\V)$.  
Hence $\V$ has unitary or finitary exact type.
\end{cor}
\begin{proof}

As $\V$ is locally finite, each finitely generated algebra in $\V$ is finite. 
In particular, any given $\A \in \mathsf{FP}(\V)$ is finite. 
But then also $\cg_e(\A)$ is finite. Hence, using Theorem~\ref{t:kerequiv}, 
$\tp(\EUAlg_{\V}(\A)) = \tp(\cg_e(\A)) \in  \{1,\omega\}$.
\end{proof}

\begin{cor}\label{Cor:TypeUnitary}
Let $\V$ be an equational class and consider $\A \in \mathsf{FP}(\V)$ such that 
 $\cg(\A)$ is totally ordered by inclusion. 
If $\EUAlg_{\V}(\A) \neq \emptyset$, then it is totally ordered and $\tp(\EUAlg_{\V}(\A))\in\{1,0\}$. 
In particular, if $\A$  is simple, then either  $\EUAlg_{\V}(\A) = \emptyset$ or $\tp(\EUAlg_{\V}(\A))=1$.
\end{cor}


\section{Case Studies} \label{s:case_studies}

Any unitary equational class (e.g., the class of Boolean algebras or the class of all algebras for some language) 
also has unitary exact type. Similarly, any finitary equational class will have unitary or finitary exact type. 
In particular, the class of Heyting algebras is finitary~\cite{Ghi99} and also has finitary exact type: 
the identity $x \lor y \eq \top$ has two most general exact unifiers as described in Example~\ref{e:equnif}. 
Minor changes to the original proofs that the classes of groups (see~\cite{AL94}) have infinitary unification type 
and modal algebras for the logic $\lgc{K}$ (see~\cite{Jer13}) have nullary unification type establish that the 
 same holds also for the exact types. The class of semigroups has infinitary unification type~\cite{Plo72} and, by considering again 
the set $\{x \cdot y \eq y \cdot x\}$, infinitary or nullary exact type; we have been unable so far to determine which of these holds, 
however.

Below we consider some more interesting cases where the type is known to change, collecting these results in Table~\ref{table}.

\begin{exa}[Distributive Lattices]
The class $\mathcal{D}$ of distributive lattices 
has nullary unification type~\cite{Ghi97}, but unitary exact type as all finitely presented distributive lattices are exact 
(see, e.g.,~\cite[Lemma~18]{CM13}). Similarly, the classes of bounded distributive lattices~\cite{Ghi97}, 
idempotent semigroups~\cite{Ba86}, De~Morgan algebras~\cite{BC201X}, and Kleene algebras~\cite{BC201X} 
are nullary, but because these classes are locally finite, they have at most -- and indeed, it can be shown via suitable cases, 
precisely -- finitary exact type.

In such cases, we may be able to obtain characterizations and algorithms for most exact unifiers of finite sets of identities. 
Consider again the class $\mathcal{D}$ of distributive lattices, where $\lang$ is the language of lattices and 
$\alg{2}=(\{0,1\},\wedge,\vee)$ is the two-element 
distributive lattice with $0< 1$. For each set of variables $X$ and each map $g\colon X\to \{0,1\}$, let us denote by 
$\bar{g}\colon\Fml(X)\to \alg{2}$, the unique homomorphism extending $g$. 
Recall that $\alg{2}$ generates the variety $\mathcal{D}$. Hence for $\a,\b\in\Fml(X)$, 
$h_{\mathcal{D}}(\a)=h_{\mathcal{D}}(\b)$ iff $\bar{g}(\a)=\bar{g}(\b)$ for all maps $g\colon X\to \{0,1\}$. 
So for a  substitution $\si\colon \Fml(X)\to \Fml(Y)$, 
\begin{equation}\label{eq:kernelDistLat}
\ker(h_{\mathcal{D}}\circ\si)=\bigcap\bigl\{\ker(\bar{g}\circ \si)\mid g\colon Y \to \{0,1\}\bigr\}.
\end{equation}
Let $\Si$ be a finite set of $\lang$-identities with $\var(\Si) = \{x_1,\ldots,x_n\} = X$, and consider
\[
S=\{f\colon X\to \{0,1\}\mid \bar{f}(\a)=\bar{f}(\b) \mbox{ for each }\a\eq\b\in\Si\}.
\]
Let $f_1,\ldots,f_m$ be an enumeration of the maps in $S$, and let $\{y_1,\ldots, y_m\}=Y$ be 
a set of $m$ distinct variables. Fix
\[
\f=\bigvee\{y_i\wedge y_j\mid 1\leq i<j\leq m\}.
\]
We define a substitution $\si\colon\Fml(X)\to \Fml(Y)$ by
\[
\si(x_i)=\f\vee\bigvee\{y_j\mid f_j(x_i)=1\}.
\] 
To see that $\si$ is a $\mathcal{D}$-unifier of $\Si$, we claim that $\bar{g}\circ \si\in S$ for each 
$g\colon Y\to \{0,1\}$.
Note first that if $g$ is the constant map $0$, then clearly $\bar{g}(\si(\a))=0=\bar{g}(\si(\b))$ for each 
$\a\eq\b\in\Si$, that is, $\bar{g}\circ \si\in S$. 
If there exists $k\in\{1,\ldots,m\}$ such that $g(y_i)=1$ iff $i=k$, then for  $i\in\{1,\ldots,n\}$, 
\[
\bar{g}( \si(x_i))=\bar{g}(\f)\vee\bigvee\{\bar{g}(y_j)\mid f_j(x_i)=1\}=0\vee\bigvee\{g(y_j)\mid f_j(x_i)=1\}=f_k(x_i),
\]
that is, $\bar{g}\circ \si=\bar{f}_k\in S$.  
 Finally, if $g(y_i)=g(y_j)=1$ for some $i\neq j$, then $\bar{g}(\f)=1$, and 
$\bar{g}\circ\si(\a)=\bar{g}\circ\si(\b)=1$ for each $\a\eq\b\in\Si$. Hence $\si$ is a  $\mathcal{D}$-unifier of $\Si$ 
and, by \eqref{eq:kernelDistLat},
\[
\ker(h_{\mathcal{D}}\circ \si)=\bigcap\{\ker(\bar{g}\circ \si)\mid g\colon Y \to \{0,1\}\}=\bigcap\{\ker(\bar{f})\mid f\in S\}.
\]
To see that $\si$ is the most exact  $\mathcal{D}$-unifier of $\Si$, let $\si'\colon\Fml(X)\to \FML$ 
be a $\mathcal{D}$-unifier of $\Si$ and let $Z$ be a finite subset of $\omega$ such that $\si'(\Fml(X))\subseteq\Fml(Z)$. 
Then, given a map $g\colon Z\to \{0,1\}$, it follows that $\bar{g}\circ\si'(\a)=\bar{g}\circ\si'(\b)$ for each  $\a\eq\b\in\Si$. 
Therefore $\{\bar{g}\circ\si'\mid g\colon Z\to \{0,1\} \}\subseteq \{\bar{f}\mid f\in S\}$. Using~\eqref{eq:kernelDistLat},
\[
\ker(h_{\mathcal{D}}\circ \si')=\bigcap\{\ker(\bar{g}\circ \si')\mid g\colon Z\to \{0,1\}\}\supseteq \bigcap\{\ker(\bar{f})\mid f\in S\}=\ker(h_{\mathcal{D}}\circ \si).
\]
Hence $\si'\eleq_\mathcal{D} \si$.

\end{exa}

\begin{exa}[Pseudocomplemented Distributive Lattices]
The equational class $\mathfrak{B}_{\omega}$ of pseudocomplemented distributive 
lattices is the class of algebras $(B, \land, \lor, ^*, \bot, \top)$ such that $(B, \land, \lor, \bot, \top)$ is a 
bounded distributive lattice and $a \land b^* = a$ iff $a \land b = \bot$ for all $a,b \in B$. 
For each $n\in \mathbb{N}$, let $\B_n=(B_n,\wedge,\vee,^{*},\bot,\top)$ denote the finite Boolean algebra 
with $n$ atoms and let $\B_n'$ be the algebra obtained by adding a new top~$\top'$ to the underlying lattice 
of $\B_n$ and endowing it with the unique operation $^{*}$ making it into a pseudocomplemented 
distributive lattice. Let $\mathfrak{B}_{n}$ denote the subvariety of $\mathfrak{B}_{\omega}$ generated by $\B_n'$.
It was proved by Lee in~\cite{Lee70} that the subvariety lattice of $\mathfrak{B}_{\omega}$ is 
\[
\mathfrak{B}_{0}\subsetneq\mathfrak{B}_{1}\subsetneq\cdots \subsetneq \mathfrak{B}_{n}\subsetneq\cdots 
\subsetneq \mathfrak{B}_{\omega},
\]
where $\mathfrak{B}_{0}$ and $\mathfrak{B}_{1}$ are the varieties of Boolean algebras and  Stone algebras, respectively.
We have already observed that the variety $\mathfrak{B}_{0}$ of Boolean algebras has unitary exact type.
The case $\mathfrak{B}_{1}$ of Stone algebras is similar to  distributive lattices: $\mathfrak{B}_1$ 
has nullary unification type~\cite{Ghi97}, but all finitely presented Stone algebras are exact 
(see~\cite[Lemma~20]{CM13}), so $\mathfrak{B}_1$ has unitary exact type.

In \cite{Ghi97} it was proved that $\mathfrak{B}_{\omega}$ has nullary unification type, 
and the same result was proved in \cite{CaXX} for $\mathfrak{B}_n$ for each $n\geq2$. 
However, all these varieties are locally finite, so an application of Corollary~\ref{Cor:LocFin} proves that they 
have at most finitary exact type. 
It is easy to prove that $\{x\vee\neg x\eq \top\} \imp \{x \eq \top,\neg x\eq\top\}$ is $\mathfrak{B}_{\omega}$-admissible 
and $\mathfrak{B}_{n}$-admissible for each $n\geq 2$ and that neither
$\{x\vee\neg x\eq \top\} \imp x\eq \top$  nor 
$\{x\vee\neg x\eq \top\} \imp \neg x\eq\top$ is $\mathfrak{B}_{\omega}$-admissible or 
$\mathfrak{B}_{n}$-admissible with $n\geq 2$. So, using Proposition~\ref{prop:admred}, 
the classes $\mathfrak{B}_{\omega}$ and $\mathfrak{B}_{n}$ with $n\geq 2$ have finitary exact type. 
\end{exa}

\begin{exa}[Willard's Example] 
The following example of a locally finite equational class with infinitary unification type 
is due to R.~Willard (private communication). Consider a language with one binary 
operation, written as juxtaposition, and two constants $0$ and $1$. 
Let $\V$ be the  equational class  defined by 
\[
0x \eq x0 \eq 0, \qquad 1 x \eq 0, \qquad x(yz) \eq 0,
\]
and, for each $n \in \mathbb{N}$, associating to the left,
\[
xyz_1z_2\ldots z_n y \eq xyz_1z_2\ldots z_n 1.
\]
Then up to equivalence in $\V$, formulas have a normal form (again associating to the left)
\[
0, \quad 1, \quad \mbox{or} \quad xy_1y_2\ldots y_n
\]
where $x$ is any variable, $y_1,\ldots,y_n$  are variables or $1$, and $y_i=y_j\neq 1$ implies $i=j$. 
It is immediate that finitely generated free algebras are finite and hence that $\V$ is 
locally finite. Note also that $\{xy \eq 0\}$ has three most general exact unifiers 
\[
\si_1(x) = 1,\ \si_1(y) = y; \quad \si_2(x) = 0,\ \si_2(y) = y;\quad \si_3(x) = x,\ \si_3(y) = yz.
\]
So the exact type of $\V$ is finitary.

However, the following set of identities has infinitary unification type:
\[
\Si = \{xy \eq x1\}.
\]
Observe that $\si(x) = x$, $\si(y) = 1$ is a $\V$-unifier of $\Si$, as are, for 
each $n \in \mathbb{N}$ and distinct variables $z_1,\ldots,z_n$ different from $y$, 
\[
\si_n(x) = xyz_1 \ldots z_n, \quad \si_n(y) = y.
\]
Moreover, the set $\{\si_n \mid n \in \mathbb{N}\} \cup \{\si\}$ is a $\mu$-set for $\U{\V}{\Si}$. 

Finally, no finite set of identities has nullary unification type. To see this, it suffices to show that the set of substitutions over some finite 
set of variables $X$ preordered by $\leqnv$ contains no infinite strictly increasing chains. Intuitively, this is because applying a 
substitution to a  formula in normal form either produces a formula of greater or equal length or a formula equivalent to $0$. 
More formally, consider some substitution $\si$ over $X = \{x_1,\ldots,x_n\}$.  We prove that $\si$ does not form part of an infinite 
strictly increasing chain by induction on the number $k$ of variables $x$ in $X$ such that $\si(x)$ is equivalent to $0$. For the base 
case $k=0$, if $\si \leqnv \si'$, then the length of the normal form of $\si'(x_i)$ must  be smaller than or equal to the length of the 
normal form of $\si(x_i)$. As there are finitely many non-equivalent (up to the names of the variables) such strings of characters, there 
are finitely many non-equivalent possible $\si'$ more general than $\si$. For the inductive step, we suppose that $\si(x_i)$ is equivalent 
to $0$ and assume for a contradiction that $\si$ forms part of an infinite strictly increasing chain of substitutions. 
Observe that $\si'(x_i)$ must be equivalent to $0$ for every $\si'$ above $\si$ in the chain; otherwise, by the induction 
hypothesis applied to $\si'$, the chain is finite, contradicting our assumption.  
But then we can construct another strictly increasing 
infinite chain of substitutions extending $\si$ by setting $\si'(x_i)=z$ for a fresh variable $z$, for each $\si'$ in the original chain,  contradicting the induction hypothesis.
\end{exa}

\begin{exa}[MV-Algebras]
It was proved in~\cite{MS13}  that the equational class $\mathcal{MV}$ of MV-algebras, the algebraic semantics of \L ukasiewicz 
infinite-valued logic (see~\cite{COM99} for details), has nullary unification type. 
This class  is not locally finite, so Corollary~\ref{Cor:LocFin} does not apply, but combining results 
from~\cite{Jer09b} and~\cite{Cab}, we can still prove that it has finitary exact type. 

Let $\lang$ be the language of MV-algebras and let $\Si$ be a finite set of $\lang$-identities. 
Finitely presented MV-algebras admit a presentation  
$\{\a\eq \top\}$, so there is no loss of generality in assuming that $\Si=\{\a\eq \top\}$. Let us fix $X=\var(\Si)$ 
and $\A=\Fp_{\mathcal{MV}}(\{\a\eq\top\},X)$.
A combination of  \cite[Theorem~3.8]{Jer09b} and \cite[Theorem~4.18]{Cab} establishes that 
there exist $\b_1,\ldots,\b_n\in\Fml(X)$ such that 
\begin{enumerate}[label=(\roman*)]
\item  $\{\a\eq\top\}\imp\{\b_1\eq\top,\ldots,\b_n\eq\top\}$ is $\mathcal{MV}$-admissible;
\item $\mathcal{MV} \models \{\b_i\eq\top\} \imp \a\eq\top$ for each $i\in\{1,\ldots,n\}$;
\item $\Fp_{\mathcal{MV}}(\{\b_i\eq\top\},X)$ is exact for each $i\in\{1,\ldots,n\}$.
\end{enumerate}
Defining $\B_i=\Fp_{\mathcal{MV}}(\{\b_i\eq\top\},X)$, from (ii), we obtain that for each $i\in\{1,\ldots,n\}$, there exists a homomorphism 
$e_i\colon\A\to \B_i$ such that $\rho_{(\{\b_i\eq\top\},X,\mathcal{MV})}=e_i\circ\rho_{(\{\a\eq\top\},X,\mathcal{MV})}$. As 
$\rho_{(\{\b_i\eq\top\},X,\mathcal{MV})}$ is onto, so is $e_i$. 
By (iii), it follows that $S=\{e_1,\ldots,e_n\}$ is a set of coexact $\mathcal{MV}$-unifiers of $\A$. 

We claim now that $S$ is a complete set in $\EUAlg_{\mathcal{MV}}(\A)$. Consider 
$e\colon \A\to\alg{C}\in\EUAlg_{\mathcal{MV}}(\A)$. By~{(i)}, there exists $i\in\{1,\ldots,n\}$ and $h\colon \B_i\to \alg{C}$ 
such that $e\circ\rho_{(\{\a\eq\top\},X,\mathcal{MV})}=h\circ\rho_{(\{\b_i\eq\top\},X,\mathcal{MV})}$. 
As $\rho_{(\{\a\eq\top\},X,\mathcal{MV})}$ is onto and $\rho_{(\{\b_i\eq\top\},X,\mathcal{MV})}=e_i\circ\rho_{(\{\a\eq\top\},X,\mathcal{MV})}$, 
it follows that $e=h\circ e_i$, that is, $e\leq e_i$. 
This proves that $\tp(\EUAlg_{\mathcal{MV}}(\A))\in\{1,\omega\}$. Hence  the exact type of $\mathcal{MV}$ is either unitary or finitary. 
But also $\{x \lor \lnot x \eq \top\}$ has a $\mu$-set $\{\si_1,\si_2\}$ where $\si_1(x) = \top$ and $\si_2(x) = \bot$. (Reasoning in the 
standard MV-algebra over $[0,1]$, there are only two continuous functions $f \colon [0,1] \to [0,1]$ satisfying 
$\max(f(\lambda),1-f(\lambda)) = 1$ for each $\lambda\in[0,1]$, 
namely $f = 1$ and $f = 0$.) So $\mathcal{MV}$ has finitary exact type. 
\end{exa}

\begin{table}[t]
\begin{center}
\begin{tabular}{|@{\ \ }c@{\ \ }|@{\ \ }c@{\ \ }|@{\ \ }c|}
\hline
Equational Class			& Unification Type 	& 	Exact Type\\
\hline\hline&&\\[-.35cm]
Boolean Algebras				&	Unitary		&	Unitary\\
\hline&&\\[-.35cm]
Heyting Algebras				&	Finitary		&	Finitary\\
\hline&&\\[-.35cm]
Groups					&	Infinitary		&	Infinitary\\
\hline&&\\[-.35cm]
Semigroups					&	Infinitary		&	Infinitary or Nullary\\
\hline&&\\[-.35cm]
Modal Algebras 				&	Nullary		&	Nullary\\
\hline 
&&\\[-.35cm]
Distributive Lattices			&	Nullary		&	Unitary\\
\hline&&\\[-.35cm]
Stone Algebras				&	Nullary		&	Unitary\\		
\hline&&\\[-.35cm]
Bounded Distributive Lattices	&	Nullary		&	Finitary\\
\hline&&\\[-.35cm]
Pseudocomplemented Distributive Lattices	& Nullary	& Finitary\\
\hline&&\\[-.35cm]
Idempotent Semigroups		&	Nullary		&	Finitary\\
\hline&&\\[-.35cm]
De Morgan Algebras		&	Nullary		&	Finitary\\
\hline&&\\[-.35cm]
Kleene Algebras		&	Nullary		&	Finitary\\
\hline&&\\[-.35cm]
MV-Algebras					&	Nullary		&	Finitary\\
\hline&&\\[-.35cm]
Willard's Example					&	Infinitary		&	Finitary\\
\hline
\end{tabular}
\end{center}
\caption{Comparison of unification types and exact types}
\label{table}
\end{table}


\section{Concluding Remarks} \label{s:concluding_remarks}

In this paper, we have introduced  a new hierarchy of exact types based on an inclusion preordering of 
unifiers, and shown, via an algebraic interpretation of unifiers, that certain classes have 
nullary or infinitary unification type, but unitary or finitary  exact type. We do not know, however, if there exist 
equational classes of (i) finitary unification type that have unitary exact type, (ii) infinitary unification type 
that have unitary or nullary exact type, or (iii) nullary unification type that have infinitary exact type.

In \cite{CM13}, the current authors have presented axiomatizations for the
admissible rules of several locally finite (and hence at most finitary exact type) equational classes 
with nullary unification type. In all these cases, a complete description of exact algebras, and the 
unitary or finitary exact  type plays a central (if implicit) role. We 
expect that this approach will also be useful for addressing admissibility in other classes of algebras that have unitary or finitary 
exact type, but nullary or infinitary unification type: e.g., the locally finite equational classes of pseudocomplemented 
distributive lattices (see~\cite{CaXX}) and Sugihara algebras, the algebraic semantics of the relevant logic R-Mingle (see~\cite{Met15}). 
Note, however, that the most significant open problems in this area concern the decidability and axiomatization of 
unifiability and admissibility in the modal logic $\lgc{K}$, where the exact type remains nullary.

Finally, although it is possible, as in the case of distributive lattices above, to obtain algorithms for building 
a (finite) set of most general exact unifiers for a finite set of identities, we do not yet have a general method, 
even for locally finite equational classes. Here the problem is that we may be able to construct the congruence lattice 
of the relevant algebra but we do not know how to decide if the quotient of the algebra by a particular 
congruence embeds into the free algebra on countably infinitely many generators of the class.

\bibliographystyle{plain}

\begin{thebibliography}{10}

\bibitem{AL94}
M.~Albert and J.~Lawrence.
\newblock Unification in varieties of groups: nilpotent varieties.
\newblock {\em Canadian Journal of Mathematics}, 46:1135--1149, 1994.

\bibitem{Ba86}
F.~Baader.
\newblock The theory of idempotent semigroups is of unification type zero.
\newblock {\em Journal of Automated Reasoning}, pages 283--286, 1986.

\bibitem{BS01}
F.~Baader and W.~Snyder.
\newblock Unification theory.
\newblock In {\em Handbook of Automated Reasoning}, volume~I, chapter~8, pages
  447--533. Elsevier Science B.V., 2001.

\bibitem{BR11b}
S.~Babenyshev and V.~Rybakov.
\newblock Linear temporal logic {LTL}: Basis for admissible rules.
\newblock {\em Journal of Logic and Computation}, 21(2):157--177, 2011.

\bibitem{BR11a}
S.~Babenyshev and V.~Rybakov.
\newblock Unification in linear temporal logic {LTL}.
\newblock {\em Annals of Pure and Applied Logic}, 162(12):991--1000, 2011.

\bibitem{BC201X}
S.~Bova and L.~M. Cabrer.
\newblock {Unification and projectivity in De Morgan and Kleene algebras}.
\newblock {\em Order}, 31(2):159--187, 2014.

\bibitem{BS81}
S.~Burris and H.~P. Sankappanavar.
\newblock {\em A Course in Universal Algebra}.
\newblock Springer, 1981.

\bibitem{BS87}
W.~Buttner and H.~Simonis.
\newblock Embedding boolean expressions into logic programming.
\newblock {\em Journal of Symbolic Computation}, 4(2):191--205, 1987.

\bibitem{Cab}
L.~M. Cabrer.
\newblock Simplicial geometry of unital lattice-ordered abelian groups.
\newblock {\em Forum Mathematicum}, 27(3):1309--1344, 2015.

\bibitem{CaXX}
L.~M. Cabrer.
\newblock Unification on subvarieties of pseudocomplemented distributive
  lattices.
\newblock {\em Notre Dame Journal of Formal Logic}, (in press).

\bibitem{CM13}
L.~M. Cabrer and G.~Metcalfe.
\newblock Admissibility via natural dualities.
\newblock {\em Journal of Pure and Applied Algebra}, 219(9):4229--4253, 2015.

\bibitem{COM99}
R.~Cignoli, I.~M.~L. D'Ottaviano, and D.~Mundici.
\newblock {\em Algebraic Foundations of Many-Valued Reasoning}, volume~7 of
  {\em Trends in Logic}.
\newblock Kluwer, Dordrecht, 1999.

\bibitem{CM10}
P.~Cintula and G.~Metcalfe.
\newblock Admissible rules in the implication-negation fragment of
  intuitionistic logic.
\newblock {\em Annals of Pure and Applied Logic}, 162(10):162--171, 2010.

\bibitem{DeJongh82}
D.~H.~J. de~Jongh.
\newblock Formulas of one propositional variable in intuitionistic arithmetic.
\newblock In {\em Stud. Log. Found. Math. 110, The L. E. J. Brouwer Centenary
  Symposium}, pages 51--64. Elsevier, 1982.

\bibitem{Ghi97}
S.~Ghilardi.
\newblock Unification through projectivity.
\newblock {\em Journal of Logic and Computation}, 7(6):733--752, 1997.

\bibitem{Ghi99}
S.~Ghilardi.
\newblock Unification in intuitionistic logic.
\newblock {\em Journal of Symbolic Logic}, 64(2):859--880, 1999.

\bibitem{Ghi00}
S.~Ghilardi.
\newblock Best solving modal equations.
\newblock {\em Annals of Pure and Applied Logic}, 102(3):184--198, 2000.

\bibitem{Gor98}
V.~A. Gorbunov.
\newblock {\em Algebraic Theory of Quasivarieties}.
\newblock Springer, 1998.

\bibitem{GI14}
J.~P. Goudsmit and R.~Iemhoff.
\newblock On unification and admissible rules in {G}abbay-de {J}ongh logics.
\newblock {\em Annals of Pure and Applied Logic}, 165(2):652--672, 2014.

\bibitem{HS06}
M.~Hoche and P.~Szab{\'o}.
\newblock Essential unifiers.
\newblock {\em Journal of Applied Logic}, 4(1):1--25, 2006.

\bibitem{Iem01}
R.~Iemhoff.
\newblock On the admissible rules of intuitionistic propositional logic.
\newblock {\em Journal of Symbolic Logic}, 66(1):281--294, 2001.

\bibitem{Iem05}
R.~Iemhoff.
\newblock Intermediate logics and {{V}}isser's rules.
\newblock {\em Notre Dame Journal of Formal Logic}, 46(1):65--81, 2005.

\bibitem{Jer13}
E.~Je{\v r}{{\'a}}bek.
\newblock Blending margins: the modal logic {K} has nullary unification type.
\newblock {\em Journal of Logic and Computation} (in press,
  DOI:10.1093/logcom/ext055).

\bibitem{Jer05}
E.~Je{\v r}{{\'a}}bek.
\newblock Admissible rules of modal logics.
\newblock {\em Journal of Logic and Computation}, 15:411--431, 2005.

\bibitem{Jer09a}
E.~Je{\v r}{\'a}bek.
\newblock Admissible rules of {{\L}}ukasiewicz logic.
\newblock {\em Journal of Logic and Computation}, 20(2):425--447, 2010.

\bibitem{Jer09b}
E.~Je{\v r}{\'a}bek.
\newblock Bases of admissible rules of {{\L}}ukasiewicz logic.
\newblock {\em Journal of Logic and Computation}, 20(6):1149--1163, 2010.

\bibitem{Jer13b}
E.~Je{\v r}{{\'a}}bek.
\newblock The complexity of admissible rules of {{\L}}ukasiewicz logic.
\newblock {\em Journal of Logic and Computation}, 23(3):693--705, 2013.

\bibitem{Lee70}
K.~B. Lee.
\newblock Equational classes of distributive pseudocomplemented lattices.
\newblock {\em Canadian Journal of Mathematics}, 22:881--891, 1970.

\bibitem{MS13}
V.~Marra and L.~Spada.
\newblock Duality, projectivity, and unification in {{\L}}ukasiewicz logic and
  {MV}-algebras.
\newblock {\em Annals of Pure and Applied Logic}, 164(3):192--210, 2013.

\bibitem{Met15}
G.~Metcalfe.
\newblock An {A}vron rule for fragments of r-mingle.
\newblock {\em Journal of Logic and Computation}, (in press).

\bibitem{MR13}
G.~Metcalfe and C.~R{\"{o}}thlisberger.
\newblock Admissibility in finitely generated quasivarieties.
\newblock {\em Logical Methods in Computer Science}, 9(2), 2013.

\bibitem{OR13}
S.~Odintsov and V.~Rybakov.
\newblock Unification and admissible rules for paraconsistent minimal
  {J}ohanssons' logic {J} and positive intuitionistic logic {IPC$^{\mbox{+}}$}.
\newblock {\em Annals of Pure and Applied Logic}, 164(7-8):771--784, 2013.

\bibitem{Plo72}
G.~Plotkin.
\newblock Building in equational theories.
\newblock {\em Machine Intelligence}, 7:73--90, 1972.

\bibitem{Ryb97}
V.~Rybakov.
\newblock {\em Admissibility of Logical Inference Rules}, volume 136 of {\em
  Studies in Logic and the Foundations of Mathematics}.
\newblock Elsevier, Amsterdam, 1997.

\end{thebibliography}

\end{document}